\documentclass[twocolumn,amsthm]{autart}
\usepackage[colorlinks=true,linkcolor=black,citecolor=blue,urlcolor=blue]{hyperref}
\usepackage{amssymb,amsmath}
\usepackage{newtxtext,newtxmath}
\usepackage{amsfonts}
\usepackage{graphicx}
  \graphicspath{{./pics/}}
  \DeclareGraphicsExtensions{.pdf}
\usepackage{xspace}
\usepackage{pgf}
\usepackage{tikz}
  \usetikzlibrary{arrows,shapes,fit}
  \usetikzlibrary{automata}
  \usetikzlibrary{backgrounds}
  \usetikzlibrary{decorations.pathmorphing}
  \usetikzlibrary{calc}
\usepackage{url}
\usepackage[T1]{fontenc}
\usepackage[latin1]{inputenc}
\usepackage[english]{babel}
\usepackage{booktabs}
  \newcommand{\ra}[1]{\renewcommand{\arraystretch}{#1}}

\allowdisplaybreaks

\usepackage[textwidth=1.9cm,color=green!10,textsize=footnotesize]{todonotes}

\DeclareMathOperator{\A}{\mathcal A}

\DeclareMathOperator{\D}{\mathcal D}

\newcommand{\complclass}[1]{{\sc #1}\xspace}
\newcommand{\NL}{\complclass{NL}}

\newcommand{\NP}{\complclass{NP}}
\newcommand{\PSpace}{\complclass{PSpace}}
\newcommand{\Expspace}{\complclass{ExpSpace}}
\newcommand{\PTime}{\complclass{P}}

\journal{}

\begin{document}
\begin{frontmatter}
\title{On Verification of D\mbox{-}Detectability for Discrete Event Systems\thanksref{footnoteinfo}}

\thanks[footnoteinfo]{An extended abstract of this work was presented at the Workshop on Discrete Event Systems WODES 2020~\cite{wodes2020}.
Corresponding author: T. Masopust, tel. +420222090785, fax +420541218657.}

\author[up]{Ji{\v r}{\' i}~Balun}\ead{jiri.balun01{@}upol.cz} \and
\author[up,im]{Tom{\' a}{\v s}~Masopust}\ead{masopust{@}math.cas.cz}

\address[up]{Department of Computer Science, Faculty of Science, Palacky University, 17. listopadu 12, 771 46 Olomouc, Czechia}
\address[im]{Institute of Mathematics of the Czech Academy of Sciences}


\begin{abstract}
  Detectability has been introduced as a generalization of state-estimation properties of discrete event systems studied in the literature. It asks whether the current and subsequent states of a system can be determined based on observations. Since, in some applications, to exactly determine the current and subsequent states may be too strict, a relaxed notion of D\mbox{-}de\-tect\-ability has been introduced, distinguishing only certain pairs of states rather than all states. Four variants of D\mbox{-}de\-tect\-ability have been defined: strong (periodic) D\mbox{-}de\-tect\-ability and weak (periodic) D\mbox{-}de\-tect\-ability. Deciding weak (periodic) D\mbox{-}de\-tect\-ability is \PSpace-complete, while deciding strong (periodic) detectability or strong D\mbox{-}de\-tect\-ability is polynomial (and we show that it is actually \NL-complete). However, to the best of our knowledge, it is an open problem whether there exists a polynomial-time algorithm deciding strong periodic D\mbox{-}de\-tect\-ability. We solve this problem by showing that deciding strong periodic D\mbox{-}detectability is a \PSpace-complete problem, and hence there is no polynomial-time algorithm unless \PSpace $=$ \PTime. We further show that there is no polynomial-time algorithm deciding strong periodic D\mbox{-}de\-tect\-ability even for systems with a single observable event, unless \PTime $=$ \NP. Finally, we propose a class of systems for which the problem is tractable.
\end{abstract}

\end{frontmatter}

\section{Introduction}
  Detectability of discrete event systems (DESs) modeled by finite automata has been introduced by Shu et al.~\cite{ShuLinYing2007} as a generalization of other notions studied in the literature, including stability of Ozveren and Willsky~\cite{OzverenW1990} and observability of Caines et al.~\cite{CainesGW1988} or of Ramadge~\cite{Ramadge1986}. An evidence that many practical problems can be formulated as the detectability problem for DESs has been provided by Shu and Lin~\cite{ShuLin2011}. Furthermore, Lin~\cite{Lin11} has shown that detectability is closely related to other important properties, such as observability, diagnosability, and opacity.

  Detectability is a state-estimation property asking whether the current and subsequent states of a DES can be determined after a finite number of observations. Shu et al.~\cite{ShuLinYing2007} have defined four variants of detectability: strong (periodic) detectability and weak (periodic) detectability. In their work, they first studied detectability for deterministic DESs, which are DESs modeled by deterministic finite automata with a set of initial states. The motivation for a set of initial states rather than a single initial state results from the observation that it is often unknown which state the system is initially in. They proposed an exponential algorithm for deciding detectability of a deterministic DES based on the computation of an observer. Shortly after, Shu and Lin~\cite{ShuLin2011} extended the problem to nondeterministic DESs (DESs modeled by general nondeterministic finite automata) and designed an algorithm deciding strong (periodic) detectability of nondeterministic DESs in polynomial time. Deciding strong (periodic) detectability was later shown \NL-complete~\cite{Masopust2018b}, that is, the problem is efficiently solvable on a parallel computer.
  
  The complexity of deciding weak (periodic) detectability has been investigated only recently. Zhang~\cite{Zhang17} has shown that deciding weak (periodic) detectability is \PSpace-comp\-lete and that it remains \PSpace-hard even for deterministic DESs with all events observable. Masopust~\cite{Masopust2018b} further strengthened these results by proving the same complexity for structurally ``simplest'' deadlock-free DESs that are modeled by deterministic finite automata without non-trivial cycles. 
  
  Since the requirement in the definition of detectability to exactly determine the current and subsequent states after a finite number of observations may be too strict in some applications, Shu and Lin~\cite{ShuLin2011} relaxed the notion of detectability to a so-called {\em D\mbox{-}de\-tect\-ability} property. The idea behind the relaxation is to distinguish only certain pairs of states rather than all states of the system. Four variants of D\mbox{-}de\-tect\-ability have been defined: strong (periodic) D\mbox{-}de\-tect\-ability and weak (periodic) D\mbox{-}de\-tect\-ability.

  The notion of (D-)detectability has been extended in many directions. To mention a few, Shu and Lin~\cite{ShuLin2013} extended strong (D-)detectability to delayed (D-)detectability, motivated by discrete event systems with delays, and designed a polynomial-time algorithm to check strong \mbox{(D-)}de\-tect\-ability for delayed DESs. Zhang and Giua~\cite{zhang2019revisiting} have recently improved the algorithm for checking strong delayed \mbox{(D-)}de\-tect\-ability. They further introduced several other notions of detectability, see Zhang et al.~\cite{Zhang2020} for more details. 
  Alves and Basilio~\cite{AlvesB19} studied \mbox{(D-)}de\-tect\-ability for discrete event systems with multi-channel communication networks. 
  Yin and Lafortune~\cite{YinLafortune17} examined the verification of weak and strong detectability properties for modular DESs, and showed that checking both is \PSpace-hard. The exact complexities of these two problems have recently been resolved by Masopust and Yin~\cite{MasopustY19}. They are, respectively, \PSpace-complete and \Expspace-complete.
  We refer the reader to Hadjicostis~\cite{Hadjicostis2020} for the latest development of state-estimation properties. 
  
  Since detectability is a special case of D\mbox{-}de\-tect\-ability, deciding D\mbox{-}de\-tect\-ability is at least as hard as deciding detectability. An immediate consequence is that the complexity of deciding whether a DES satisfies weak (periodic) D\mbox{-}de\-tect\-ability is \PSpace-complete.
  
  The case of strong D\mbox{-}de\-tect\-ability is similar to that of strong detectability. For strong (periodic) detectability, Shu and Lin~\cite{ShuLin2011} designed a detector that can decide, in polynomial time, whether a DES satisfies strong (periodic) detectability. They have further shown that their detector is also suitable for deciding strong D\mbox{-}de\-tect\-ability. Consequently, the complexity of verifying whether a DES satisfies strong D\mbox{-}de\-tect\-ability is polynomial; see Zhang and Giua~\cite{zhang2019revisiting} for details on the algorithmic complexity. 
  We further improve this result by showing that deciding whether a DES satisfies strong D\mbox{-}de\-tect\-ability is \NL-comp\-lete (Theorem~\ref{thm_sdnl-c}). Since \NL is the class of problems that can be efficiently parallelized, see Arora and Barak~\cite{AroraBarak2009} for details, we obtain that the verification of strong D\mbox{-}de\-tect\-ability can be efficiently verified on a parallel computer. 
  
  However, deciding strong {\em periodic\/} D\mbox{-}de\-tect\-ability is more involved. Although the de\-tec\-tor-based technique provides a polynomial-time algorithm to decide strong periodic detectability, Shu and Lin~\cite{ShuLin2011} give an example that this algorithm does not work for checking strong periodic D\mbox{-}de\-tect\-ability. They leave the question of the existence of a polynomial-time algorithm deciding strong periodic D\mbox{-}de\-tect\-ability of a DES open. To the best of our knowledge, this question has not yet been answered in the literature.
  We answer this question by showing that there does not exist any algorithm that would decide, in polynomial time, whether a DES satisfies strong periodic D\mbox{-}de\-tect\-ability (Theorem~\ref{thm_spd-ps}), unless \PTime $=$ \PSpace. 
  The question whether \PTime $=$ \PSpace is a longstanding open problem of computer science asking whether every problem solvable in polynomial space can also be solved in polynomial time. It is generally believed that it is not the case. In particular, Theorem~\ref{thm_spd-ps} shows that the strong periodic D\mbox{-}de\-tect\-ability problem is \PSpace-complete. Formulated differently, the result says that the technique based on the computation of the observer is in principle optimal. Notice that since \NL is a strict subclass of \PSpace, strong periodic D\mbox{-}de\-tect\-ability is significantly more complex than its non-periodic counterpart---strong D\mbox{-}de\-tect\-ability. 

  We further show that strong {\em periodic\/} D\mbox{-}de\-tect\-ability is more complex than strong D\mbox{-}de\-tect\-ability even for systems having only a single observable event. Namely, we show that strong periodic D\mbox{-}de\-tect\-ability cannot be verified in polynomial time even for DESs that have only a single observable event (Theorem~\ref{np-c}), unless \PTime $=$ \NP.

  Finally, we specify a class of systems for which deciding strong {\em periodic\/} D\mbox{-}de\-tect\-ability is in polynomial time. Namely, we consider the class of systems modeled by NFAs where all cycles are in the form of self-loops and where there is no nondeterministic choice between a step changing the state and a step not changing the state under the same observation. These restrictions are purely structural, and the models are called {\em rpoDES}, see Section~\ref{secRpoNFAs} for details.

  Our contributions, compared with known results, are summarized in Tables~\ref{tableKnown} and~\ref{table2}.
  \begin{table*}\centering
    \ra{1.2}
    \begin{tabular}{@{}lllllllll@{}}\toprule
                      & \multicolumn{3}{l}{DESs} 
                      & \phantom{a}
                      & \multicolumn{3}{l}{rpoDESs} \\\cmidrule{2-4}\cmidrule{6-8}
                      & detectability     
                      & \multicolumn{2}{l}{D\mbox{-}de\-tect\-ability}
                      & 
                      & detectability     
                      & \multicolumn{2}{l}{D\mbox{-}de\-tect\-ability}\\
                      \cmidrule{2-4}\cmidrule{6-8}
                      & known 
                      & known
                      & new
                      & 
                      & known
                      & known
                      & new\\
      \midrule
      strong          & \NL-c~\cite{Masopust2018b}      
                      & in \PTime~\cite{ShuLin2011} 
                      & \NL-c (Thm~\ref{thm_sdnl-c})
                      & 
                      & \NL-c~\cite{Masopust2018b}
                      & in \PTime~\cite{ShuLin2011} 
                      & \NL-c (Thm~\ref{thm_sdnl-c} \& Cor~\ref{cor1})\\
      weak            & \PSpace-c~\cite{Zhang17}  
                      & \PSpace-c 
                      & ---
                      & 
                      & \PSpace-c~\cite{Masopust2018b}
                      & \PSpace-c
                      & ---\\
      strong periodic & \NL-c~\cite{Masopust2018b}      
                      & ?
                      & \PSpace-c (Thm~\ref{thm_spd-ps})
                      & 
                      & \NL-c~\cite{Masopust2018b}
                      & ?
                      & \NL-c (Thm~\ref{thmRpoNFAs}) \\
      weak periodic   & \PSpace-c~\cite{Zhang17}
                      & \PSpace-c 
                      & ---
                      & 
                      & \PSpace-c~\cite{Masopust2018b}
                      & \PSpace-c
                      & ---\\
      \bottomrule
    \end{tabular}
    \medskip
    \caption{Summary of known and new results for DESs and rpoDESs; ? means that the problem was open; results easily derivable from existing results are also placed among known results}
    \label{tableKnown}
  \end{table*}

  \begin{table*}\centering
    \ra{1.2}
    \begin{tabular}{@{}llllll@{}}\toprule
                      & \multicolumn{2}{l}{unary DESs}
                      & \phantom{a}
                      & \multicolumn{2}{l}{unary rpoDESs}\\
                      \cmidrule{2-3}\cmidrule{5-6}
                      & detectability
                      & D\mbox{-}de\-tect\-ability
                      &
                      & detectability
                      & D\mbox{-}de\-tect\-ability \\
        \midrule
      strong          & \NL-c~\cite{Masopust2018b}
                      & \NL-c (Thm~\ref{thm_sdnl-c})
                      &
                      & \NL-c (Cor~\ref{cor1})
                      & \NL-c (Cor~\ref{cor1})\\
      weak            & \NL-c (Thm~\ref{thm6})
                      & \NL-c (Thm~\ref{thm6})
                      &
                      & \NL-c (Thm~\ref{thm6}) 
                      & \NL-c (Thm~\ref{thm6}) \\
      strong periodic & \NL-c~\cite{Masopust2018b}
                      & \NP-c (Thm~\ref{np-c})
                      &
                      & \NL-c (Thm~\ref{thmRpoNFAs})
                      & \NL-c (Thm~\ref{thmRpoNFAs}) \\
      weak periodic   & \NL-c (Thm~\ref{thm6})
                      & \NP-c (Thm~\ref{thm6})
                      &
                      & \NL-c (Thm~\ref{thm6})
                      & \NL-c (Thm~\ref{thm6}) \\
      \bottomrule
    \end{tabular}
    \caption{Summary of known and new results for DESs and rpoDESs with a single observable event}
    \label{table2}
  \end{table*}

\section{Preliminaries and Definitions}
  For a set $A$, $|A|$ denotes the cardinality of $A$ and $2^{A}$ its power set. An alphabet $\Sigma$ is a finite nonempty set of events. A string over $\Sigma$ is a sequence of events of $\Sigma$. Let $\Sigma^*$ denote the set of all finite strings over $\Sigma$; the empty string is denoted by $\varepsilon$. For a string $u \in \Sigma^*$, $|u|$ denotes its length. As usual, the notation $\Sigma^+$ stands for $\Sigma^*\setminus\{\varepsilon\}$. 

  A {\em nondeterministic finite automaton\/} (NFA) over an alphabet $\Sigma$ is a structure $\A = (Q,\Sigma,\delta,I,F)$, where $Q$ is a finite set of states, $I\subseteq Q$ is a set of initial states, $F \subseteq Q$ is a set of marked states, and $\delta \colon Q\times\Sigma \to 2^Q$ is a transition function that can be extended to the domain $2^Q\times\Sigma^*$ by induction. The {\em language recognized by $\A$\/} is the set $L(\A) = \{w\in \Sigma^* \mid \delta(I,w)\cap F \neq\emptyset\}$. 
  Equivalently, the transition function $\delta$ is a relation $\delta \subseteq Q\times \Sigma \times Q$, where, for instance, $\delta(q,a)=\{s,t\}$ denotes the two transitions $(q,a,s)$ and $(q,a,t)$.

  The NFA $\A$ is {\em deterministic\/} (DFA) if it has a unique initial state, i.e., $|I|=1$, and no nondeterministic transitions, i.e., $|\delta(q,a)|\le 1$ for every $q\in Q$ and $a \in \Sigma$. The DFA $\A$ is {\em total\/} if in every state, a transition under every event is defined, i.e., $|\delta(q,a)|=1$ for every $q\in Q$ and $a \in \Sigma$. For DFAs, we identify singletons with their elements and simply write $p$ instead of $\{p\}$. Specifically, we write $\delta(q,a)=p$ instead of $\delta(q,a)=\{p\}$.

  A {\em discrete event system\/} (DES) is an NFA $G$ with all states marked. Hence we simply write $G=(Q,\Sigma,\delta,I)$ leaving out the set of marked states. Additionally, the alphabet $\Sigma$ is partitioned into the set $\Sigma_o$ of {\em observable events\/} and the set $\Sigma_{uo}=\Sigma\setminus\Sigma_o$ of {\em unobservable events}. 
  
  State-estimation properties are based on the observation of events. The observation is described by projections. The projection $P\colon \Sigma^* \to \Sigma_o^*$ is a morphism defined by $P(a) = \varepsilon$ for $a\in \Sigma\setminus \Sigma_o$, and $P(a)= a$ for $a\in \Sigma_o$. The action of $P$ on a string $w=a_1 a_2 \cdots a_n$, where $a_i \in \Sigma$ for $1\le i\le n$, is to erase all events from $w$ that do not belong to $\Sigma_o$; in particular, $P(a_1 a_2\cdots a_n)=P(a_1) P(a_2) \cdots P(a_n)$. The definition can readily be extended to infinite strings and languages.
  
  Shu and Lin~\cite{ShuLin2011} make the following two reasonable assumptions on the DES $G=(Q,\Sigma,\delta,I)$ that we adopt:
  \begin{enumerate}\label{AS}
    \item $G$ is {\em deadlock free} -- it means that for every state of the system, at least one event can occur; formally, for every $q\in Q$, there exists $\sigma \in \Sigma$ such that $\delta(q,\sigma)\neq\emptyset$.
  
    \item No loop in $G$ consists solely of unobservable events -- for every $q\in Q$ and every $w \in \Sigma_{uo}^+$, $q\notin \delta(q,w)$.
  \end{enumerate}
  We point out that to verify whether a system satisfies these two properties is very easy. The violation of any of the properties is often considered a modeling error. Moreover, omitting the conditions does not change our results. 

  The set of infinite sequences of events (or trajectories) generated by the DES $G$ is denoted by $L^\omega (G)$. 
  Given a set $Q' \subseteq Q$, the set of all possible states after observing a string $t \in \Sigma_o^*$ is denoted by 
  \[
    R(Q',t) = \bigcup_{w \in \Sigma^*, P(w) = t} \delta(Q',w)\,.
  \]
  For $w \in L^\omega (G)$, we denote the set of its prefixes by $Pr(w)$.

\subsection{A Brief Complexity Review}
  We now briefly review the basics of complexity theory needed to understand the results. A {\em decision problem\/} is a yes-no question. A decision problem is {\em decidable\/} if there exists an algorithm that can solve the problem. Complexity theory classifies decidable problems to classes according to the time or space an algorithm needs to solve the problem. The complexity classes we consider in this paper are \NL, \PTime, \NP, and \PSpace. They denote the classes of problems that are solvable by a nondeterministic logarithmic-space, deterministic polynomial-time, nondeterministic polynomial-time, and deterministic polynomial-space algorithm, respectively. The hierarchy of classes is \NL $\subseteq$ \PTime $\subseteq$ \NP $\subseteq$ \PSpace. Which of the inclusions are strict is a longstanding open problem in computer science. The widely accepted conjecture is that all inclusions are strict. However, so far only the inclusion \NL $\subseteq$ \PSpace is known to be strict. A decision problem is \NL-complete (resp. \NP-complete, \PSpace-complete) if it belongs to \NL (resp. \NP, \PSpace) and every problem from \NL (resp. \NP, \PSpace) can be reduced to it by a deterministic logarithmic-space (resp. polynomial-time) algorithm.

\section{Definitions of the D-Detectability Problems}
  Shu and Lin~\cite{ShuLin2011} defined D\mbox{-}de\-tect\-ability as a generalization of detectability by making the states that need to be distinguished explicit. 
 
  Let $G=(Q,\Sigma,\delta,I)$ be a DES, and let $T_{spec}\subseteq Q\times Q$ be a relation on the set of states of $G$. The relation $T_{spec}$ specifies pairs of states that must be distinguished, and is therefore called a {\em specification}. The idea behind the definition of D\mbox{-}de\-tect\-ability is to ensure that the pairs of states from $T_{spec}$ are distinguished after a finite number of observations.
  We now recall the definitions of the four variants of D\mbox{-}de\-tect\-ability.
    
  A DES $G=(Q,\Sigma,\delta,I)$ is {\em strongly D\mbox{-}de\-tect\-able\/} with respect to projection $P\colon \Sigma^* \to \Sigma_{o}^*$ and a specification $T_{spec}$ if, for all trajectories of the system, the pairs of states of $T_{spec}$ can be distinguished in every step of the system after a finite number of observations. This is formally defined as follows:
  \begin{multline*}
    (\exists n \in \mathbb{N})(\forall s \in L^{\omega} (G))(\forall t \in Pr(s))\, 
    |P (t)| > n \\ \Rightarrow  (R(I,P(t))\times R(I,P(t))) \cap T_{spec} = \emptyset \,.
  \end{multline*}

  A DES $G=(Q,\Sigma,\delta,I)$ is {\em weakly D\mbox{-}de\-tect\-able\/} with respect to projection $P\colon \Sigma^* \to \Sigma_{o}^*$ and a specification $T_{spec}$ if, for some trajectories of the system, the pairs of states of $T_{spec}$ can be distinguished in every step of the system after a finite number of observations. This is formally defined as follows:
  \begin{multline*}
    (\exists n \in \mathbb{N})(\exists s \in L^{\omega}(G))(\forall t \in Pr(s)) \, |P (t)| > n \\ 
    \Rightarrow  ( R(I,P(t)) \times R(I,P(t)) ) \cap T_{spec} = \emptyset \,.
  \end{multline*}

  A DES $G=(Q,\Sigma,\delta,I)$ is {\em strongly periodically D\mbox{-}de\-tect\-able\/} with respect to projection $P\colon \Sigma^* \to \Sigma_{o}^*$ and a specification $T_{spec}$ if the pairs of states of $T_{spec}$ can be periodically distinguished for all trajectories of the system. Formally,
  \begin{multline*}
    (\exists n \in \mathbb{N})(\forall s \in L^{\omega}(G))(\forall t \in Pr(s))(\exists t' \in \Sigma^*)\, \\
    tt'\in Pr(s) \,\land\, |P(t')| < n \\ 
    \,\land\, ( R(I,P(tt')) \times R(I,P(tt')) ) \cap T_{spec} = \emptyset \,.
  \end{multline*}

  A DES $G=(Q,\Sigma,\delta,I)$ is {\em weakly periodically D\mbox{-}de\-tect\-able\/} with respect to projection $P\colon \Sigma^* \to \Sigma_{o}^*$ and a specification $T_{spec}$ if the pairs of states of $T_{spec}$ can be periodically distinguished for some trajectories of the system. Formally,
  \begin{multline*}
    (\exists n \in \mathbb{N})(\exists s \in L^{\omega}(G))(\forall t \in Pr(s))(\exists t' \in \Sigma^*)\, \\
    tt'\in Pr(s) \,\land\, |P(t')| < n \\ 
    \,\land\, ( R(I,P(tt')) \times R(I,P(tt')) ) \cap T_{spec} = \emptyset \,.
  \end{multline*}

\section{Results}
  We now discuss the complexity of deciding whether a DES satisfies D\mbox{-}de\-tect\-ability.
  As already pointed out in the introduction, the complexity of checking whether a DES satisfies weak (periodic) D\mbox{-}de\-tect\-ability follows directly from the complexity of checking weak (periodic) detectability. Indeed, a polynomial space is sufficient for an algorithm based on the inspection of states in the observer and works for all the D\mbox{-}de\-tect\-ability variants. Therefore, deciding weak (periodic) D\mbox{-}de\-tect\-ability is in \PSpace. On the other hand, detectability is a special case of D\mbox{-}de\-tect\-ability for $T_{spec}=Q\times Q \setminus \{(q,q) \mid q\in Q\}$. Therefore, deciding weak (periodic) D\mbox{-}de\-tect\-ability is at least as hard as deciding weak (periodic) detectability. Since the latter is \PSpace-hard, so is the former.

\subsection{Verification of Strong D-Detectability}
  Shu and Lin~\cite{ShuLin2011} designed an algorithm that verifies strong (periodic) detectability in polynomial time. Their algorithm is based on the construction of a finite automaton called a {\em detector}. Intuitively, given a DES $G$, their detector $G_{det}$ is constructed from $G$ so that
  (i) the set of initial states of $G_{det}$ is the set of all states of $G$ reachable from the initial states of $G$ under strings consisting only of unobservable events,
  (ii) all the other states of $G_{det}$ are one- or two-element subsets of the set of states of $G$, and
  (iii) the transition relation of $G_{det}$ is constructed in the similar way as that of the observer, but if the reached state $X$ in the observer consists of more than two states, then the detector $G_{det}$ has several transitions each leading to a two-element subset of $X$, see Shu and Lin~\cite{ShuLin2011} for details.
  Since the states of the detector are one- or two-element subsets, their number is polynomial. 
   
  Shu and Lin~\cite{ShuLin2011} showed that a DES $G$ satisfies strong (periodic) detectability if and only if any state reachable from any loop in $G_{det}$ consists solely (periodically) of distinguishable states.
  They further proved that their algorithm, respectively the detector, works for checking whether a DES satisfies strong D\mbox{-}de\-tect\-ability. This in particular implies that the complexity of verifying whether a DES satisfies strong D\mbox{-}de\-tect\-ability is polynomial. Zhang and Giua~\cite{zhang2019revisiting} recently improved the algorithmic complexity of this problem. 
  
  We now discuss the computational complexity of deciding strong D\mbox{-}de\-tect\-ability and show that it is an \NL-complete problem. Consequently, since \NL is the class of problems that can be efficiently parallelized, see Arora and Barak~\cite{AroraBarak2009} for details, our result shows that the question whether a DES satisfies strong D\mbox{-}de\-tect\-ability can be efficiently verified on a parallel computer. 
  
  \begin{thm}\label{thm_sdnl-c}
    Deciding whether a DES is strongly D\mbox{-}detect\-able is an \NL-complete problem.
  \end{thm}
  \begin{proof}
    We prove membership of the problem in \NL by giving a nondeterministic logarithmic-space algorithm that checks whether the condition does not hold. Since \NL is closed under complement, see Immerman~\cite{Immerman88} or Szelepcs{\'{e}}nyi~\cite{Szelepcsenyi87} for details, it shows that there is a nondeterministic logarithmic-space algorithm checking whether the condition is satisfied.
  
    To check that the property is not satisfied, our \NL algorithm guesses two states of $G_{det}$, say $x$ and $y$, where $y$ contains indistinguishable states, and verifies that
      (i)   $y$ is reachable from $x$,
      (ii)  $x$ is reachable from the initial state of $G_{det}$, and
      (iii) $x$ is in a cycle, i.e., $x$ is reachable from $x$ by a path having at least one transition.
    Notice that our algorithm does not construct the detector $G_{det}$. It only stores a constant number of states of $G_{det}$ and computes the required transitions of $G_{det}$ on demand. Therefore, our algorithm does not need more than a logarithmic space. For more details how to check reachability in \NL, we refer the reader to Masopust~\cite{Masopust2018}.
    
    \begin{figure}
      \centering
      \begin{tikzpicture}[baseline,auto,->,>=stealth,shorten >=1pt,node distance=1.5cm,thick, state/.style={circle,minimum size=7mm,inner sep=1pt,thick,draw=black,initial text=}]
        \node[state,initial above]  (1) {$s$};
        \node[state]          (2) [above of=1]  {$p$};
        \node                 (c) at ($(1)!0.5!(2)$) {};
        \node[state]          (5) [right of=2]  {$r$};
        \node[state]          (3) [right of=1,node distance=2.7cm]  {$t$};
        \node[state]          (4) [left of=c,node distance=2.7cm]  {$x$};
        \node at (2.8,1.6) {$G$};
        \path
          (2) edge node[sloped,above] {$a$} (5)
          (1) edge[bend left=10] node[pos=0.5,sloped,above] {$a$} (4)
          (2) edge[bend right=10] node[pos=0.5,sloped,above] {$a$} (4)
          (4) edge[loop above] node[pos=0.5,sloped,above] {$a$} (4)
          (3) edge[loop above] node[pos=0.5,sloped,above] {$a$} (3)
          (5) edge[bend left=11] node[pos=0.7,sloped,above] {$a$} (4)
          (1) edge[style={decorate, decoration={snake,amplitude=.4mm,segment length=1.7mm,post length=1.3mm}}] node{?} (3) ;
          ;
        \begin{pgfonlayer}{background}
          \path (2.north -| 3.east) + (0.1,0.1)    node (a) {};
          \path (1.south -| 1.west) + (-0.3,-0.1)  node (b) {};
          \path[rounded corners,draw=black,thick] (a) rectangle (b);
        \end{pgfonlayer}
      \end{tikzpicture}
      \caption{The DES $\A$ constructed from $G$ in the \NL-hardness proof of Theorem~\ref{thm_sdnl-c}}
      \label{fig1}
    \end{figure}
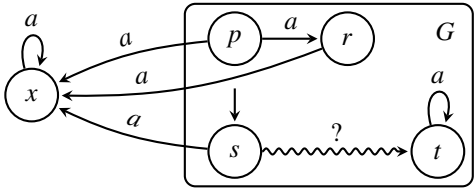    
  
    To show \NL-hardness, we reduce the {\em DAG non-reachability problem}, see Cho and Huynh~\cite{ChoH91} for details: Given a directed acyclic graph $G=(V,E)$ and two nodes $s,t\in V$, it asks whether $t$ is not reachable from $s$.
    From $G$, we construct a DES $\A=(V\cup\{x\},\{a\},\delta,s)$, where $x\notin V$ is a new state and $a$ is an observable event. For every edge $(p,r)\in E$, we add the transition $(p,a,r)$ to $\delta$, and for every $p\in V\setminus\{t\}$, we add the transition $(p,a,x)$ to $\delta$. Moreover, we add the self-loop transitions $(x,a,x)$ and $(t,a,t)$ to $\delta$. The construction is depicted in Fig.~\ref{fig1}. Notice that $\A$ is deadlock-free and has no unobservable events. 
    Let the specification $T_{spec}$ be defined as $T_{spec} = \{(t,x)\}$. We now show that $t$ is not reachable from $s$ in the graph $G$ if and only if the DES $\A$ is strongly D\mbox{-}de\-tect\-able.
    If $t$ is not reachable from $s$ in $G$, then, for every $k\ge |V|$, $\delta(s,a^k)=\{x\}$. Therefore, $\A$ is strongly D\mbox{-}de\-tect\-able.
    If $t$ is reachable from $s$ in $G$, then, for every $k\ge |V|$, $\delta(s,a^k)=\{t,x\}$. Therefore, $\A$ is not strongly D\mbox{-}de\-tect\-able, because $(t,x)\in T_{spec}$.
  \end{proof}

\subsection{Verification of Strong Periodic D-Detectability}  
  Although the de\-tec\-tor-based technique leads to a poly\-no\-mi\-al-time algorithm deciding whether a DES satisfies strong periodic detectability, Shu and Lin~\cite{ShuLin2011} have shown that this algorithm does not work for checking strong periodic D\mbox{-}de\-tect\-ability. 
  To give the reader an idea of the detector-based poly\-no\-mi\-al-time algorithm and of the problem why it does not work for checking strong periodic D\mbox{-}de\-tect\-ability, we slightly elaborate the example of Shu and Lin~\cite{ShuLin2011}.

  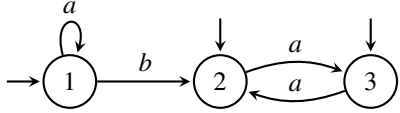
\begin{figure}
    \centering
    \begin{tikzpicture}[baseline,auto,->,>=stealth,shorten >=1pt,node distance=2cm,thick, state/.style={circle,minimum size=7mm,inner sep=1pt,thick,draw=black,initial text=}]
      \node[state,initial]  (1) {$1$};
      \node[state,initial above]  (2) [right of=1]  {$2$};
      \node[state,initial above]  (3) [right of=2]  {$3$};
      \path
        (1) edge[loop above] node[sloped,above] {$a$} (1)
        (1) edge node[sloped,above] {$b$} (2)
        (2) edge[bend left=20] node[sloped,above] {$a$} (3)
        (3) edge[bend left=20] node[sloped,above] {$a$} (2)
      ;
      \end{tikzpicture}
    \caption{The DES $G$ from Example~\ref{ex1}}
    \label{fig2}
  \end{figure}

  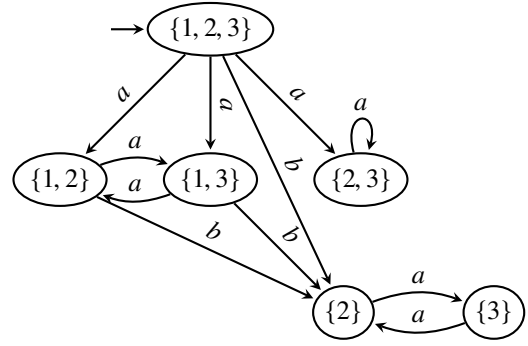
\begin{figure}
    \centering
    \begin{tikzpicture}[baseline,auto,->,>=stealth,shorten >=1pt,node distance=2cm,thick, state/.style={ellipse,minimum size=7mm,inner sep=1pt,thick,draw=black,initial text=}]
      \node[state,initial]  (1) {$\{1,2,3\}$};
      \node[state]  (3) [below of=1]  {$\{1,3\}$};
      \node[state]  (2) [left  of=3]  {$\{1,2\}$};
      \node[state]  (4) [right  of=3]  {$\{2,3\}$};
      \node[state]  (5) [below right of=3,node distance=2.5cm]  {$\{2\}$};
      \node[state]  (6) [right of=5,node distance=2cm]  {$\{3\}$};
      \path
        (1) edge node[sloped,above] {$a$} (2)
        (1) edge node[sloped,above] {$a$} (3)
        (1) edge node[sloped,above] {$a$} (4)
        (4) edge[loop above] node[sloped,above] {$a$} (4)
        (2) edge[bend left=20] node[sloped,above] {$a$} (3)
        (3) edge[bend left=20] node[sloped,above] {$a$} (2)
        (1) edge node[sloped,above] {$b$} (5)
        (2) edge node[sloped,above] {$b$} (5)
        (3) edge node[sloped,above] {$b$} (5)
        (5) edge[bend left=20] node[sloped,above] {$a$} (6)
        (6) edge[bend left=20] node[sloped,above] {$a$} (5)
      ;
      \end{tikzpicture}
    \caption{Detector $G_{det}$ constructed from the DES $G$ of Example~\ref{ex1}}
    \label{figDet}
  \end{figure}

  \begin{exmp}\label{ex1}
    Let $G=(\{1,2,3\},\{a,b\},\delta,\{1,2,3\})$ be the DES depicted in Fig.~\ref{fig2}, where both events are observable. Let the specification $T_{spec}=\{(1,3)\}$. The detector $G_{det}$ is depicted in Fig.~\ref{figDet}.
    In $G_{det}$, we can immediately see that $G$ is not strongly D\mbox{-}de\-tect\-able, since there is an infinite path that goes periodically through state $\{1,3\}$, violating thus strong D\mbox{-}de\-tect\-ability.
    From the same infinite path going periodically through the states $\{1,2\}$ and $\{1,3\}$, the reader could get an impression that $G$ is strongly periodically D\mbox{-}de\-tect\-able. However, this is not the case as can be seen from the observer depicted in Fig.~\ref{fig3}. There, there is an infinite trajectory $a^\omega$ in state $\{1,2,3\}$ that violates strong periodic D\mbox{-}de\-tect\-ability.
    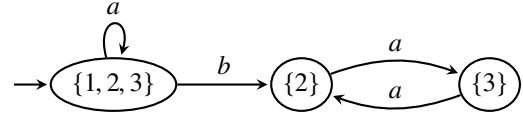
\begin{figure}
      \centering
      \begin{tikzpicture}[baseline,auto,->,>=stealth,shorten >=1pt,node distance=2.5cm,thick, state/.style={ellipse,minimum size=7mm,inner sep=1pt,thick,draw=black,initial text=}]
        \node[state,initial]  (1) {$\{1,2,3\}$};
        \node[state]  (2) [right of=1]  {$\{2\}$};
        \node[state]  (3) [right of=2]  {$\{3\}$};
        \path
          (1) edge[loop above] node[sloped,above] {$a$} (1)
          (1) edge node[sloped,above] {$b$} (2)
          (2) edge[bend left=20] node[sloped,above] {$a$} (3)
          (3) edge[bend left=20] node[sloped,above] {$a$} (2)
        ;
        \end{tikzpicture}
      \caption{The observer of the DES $G$ from Example~\ref{ex1}}
      \label{fig3}
    \end{figure}
  \hfill$\diamond$\end{exmp}
  
  Shu and Lin~\cite{ShuLin2011} have left the question whether there exists a polynomial-time algorithm deciding strong periodic D\mbox{-}de\-tect\-ability of a DES open. To the best of our knowledge, this question has not yet been answered in the literature. We answer this question in the sequel.
  We distinguish two cases based on the number of observable events in the system: 
  \begin{description}
    \item[(i)] The general case where the system has two or more observable events;
    \item[(ii)] A special case where the system  has only a single observable event.
  \end{description}

\subsection*{The case of two or more observable events}
  As pointed out above, the problem whether a DES satisfies strong or weak (periodic) D\mbox{-}de\-tect\-ability is in \PSpace. In this section, we show that deciding strong periodic D\mbox{-}de\-tect\-ability is \PSpace-hard, and hence \PSpace-complete. Consequently, there is no algorithm solving this problem in polynomial time, unless \PTime $=$ \PSpace. 

  \begin{thm}\label{thm_spd-ps}
      Deciding whether a DES is strongly periodically D\mbox{-}de\-tect\-able is a \PSpace-complete problem. The problem is \PSpace-hard even if the DES has only two observable and no unobservable events.
  \end{thm}
  \begin{proof}
    Membership in \PSpace follows from the inspection of states of the observer that are built on demand~\cite{ShuLin2011,Zhang17,Masopust2018}. 

    To show \PSpace-hardness, we reduce the intersection emptiness problem. The problem is \PSpace-complete~\cite{Garey1990} and asks, given a sequence $\A_1,\ldots,\A_n$ of total DFAs over a common alphabet $\Sigma$ with $|\Sigma|\ge 2$, whether the language $\cap_{i=1}^{n} L(\A_i)$ is empty? Without loss of generality, we may assume that $\Sigma=\{0,1\}$.
    From $\A_1,\ldots,\A_n$, we construct a DES $G$ that is strongly periodically D\mbox{-}de\-tectable if and only if the intersection of the languages of $\A_1,\ldots,\A_n$ is empty. 

    The main idea of our proof is to construct $G$ as a nondeterministic union of the automata $\A_1,\ldots,\A_n$ together with $n+1$ new states such that all and only these states are reachable at the same time if and only if the intersection is nonempty. In the case the intersection is empty, only a strict subset of the new states can be reached at the same time. After reaching the new states, the computation remains in the new states. The new states (up to one special state) form a cycle, and hence, during any further computation, the current states are periodically rotated. This allows us to make one of the new states periodically indistinguishable from the special state, and one periodically distinguishable. 
    
    Formally, let $\A_1, \ldots, \A_n$ be total DFAs over a common alphabet $\Sigma$, and let $\A_i = (Q_i,\Sigma,\delta_i,q_{0,i},F_i)$. Without loss of generality, we may assume that the states of the DFAs are pairwise disjoint. We construct a DES $G$ as a nondeterministic union of the automata $\A_i$, i.e., $G$ contains all states and transitions of every $\A_i$, and we add $n+1$ new states $q^-$, $q_1^+,\ldots,q_n^+$ and several new transitions under a new event $a\notin\Sigma$ as depicted in Fig.~\ref{fig4}. Namely, for $i=1,\ldots,n$, we add the transition $(q,a,q^-)$ for every non-marked state $q\in Q_i \setminus F_i$, and the transition $(q,a,q_i^+)$ for every marked state $q\in F_i$. Furthermore, we add the self-loop $(q^-,\sigma,q^-)$ for every $\sigma \in \Sigma' = \Sigma \cup \{a\}$. Finally, we create a cycle on the states $Q^+=\{q_1^+,\ldots,q_n^+\}$ by adding, for every $\sigma\in\Sigma'$, the transitions $(q_i^+,\sigma,q_{i+1}^+)$, for $1\le i < n$, and the transition $(q_n^+,\sigma,q_{1}^+)$.
    The set of initial states of $G$ is the set $I=\{q^-,q_{0,1},\ldots, q_{0,n}\}$ of initial states of the automata $\A_i$ plus the newly added state $q^-$. The alphabet of $G$ is $\Sigma'=\Sigma\cup\{a\}=\{0,1,a\}$, all events of which are observable.
    
    To show that the problem is \PSpace-hard for $|\Sigma'|=2$, we modify $G$ by encoding the events of $\Sigma'$ in binary as follows. Let $b$ be a new event, and let $f\colon \Sigma^* \to\{a,b\}^*$ be a morphism defined by $f(0)=ba$ and $f(1)=bb$. Now, in $G$, we replace each transition $t=(p,0,q)$ with two transitions $(p,b,p_t)$ and $(p_t,a,q)$, where $p_t$ is a new state. Similarly, we replace each transition $r=(p,1,q)$ with two transitions $(p,b,p_r)$ and $(p_r,b,q)$, where $p_r$ is a new state; see Fig.~\ref{fig5} for an illustration how to replace the transitions of the DFAs of Fig.~\ref{fig5a}. Notice that this replacement requires to add a new state for each transition of $G$, which can indeed be done in polynomial time. This results in a DES $G'$ with the alphabet $\Sigma'=\{a,b\}$, where both $a$ and $b$ are observable.
    
    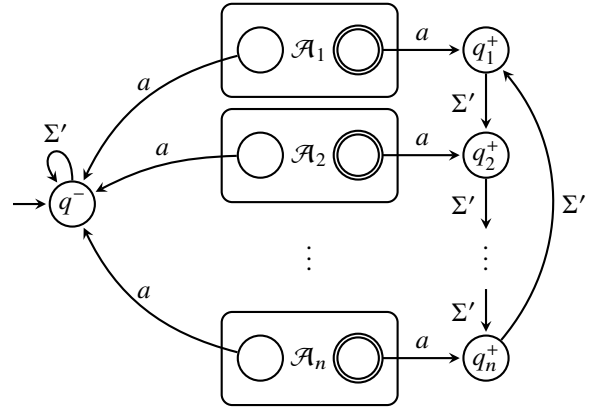
\begin{figure}
      \centering
      \begin{tikzpicture}[thick,baseline,auto,->,>=stealth,shorten >=1pt,node distance=1.3cm,state/.style={ellipse,minimum size=6mm,thick,inner sep=0pt,draw=black,initial text=},
        every node/.style={font=}]
        \node[state]  (1) {};
        \node[state,accepting]  (2) [right of=1]  {};
        \node[state]  (3) [right of=2, node distance=1.7cm]  {$q_1^+$};
        
        \node[state]  (4) [below of=1, node distance=1.4cm]  {};
        \node[state,accepting]  (5) [right of=4]  {};
        \node[state]  (6) [right of=5, node distance=1.7cm]  {$q_2^+$};

        \node[state]  (7) [below of=4,node distance=2.7cm]  {};
        \node[state,accepting]  (8) [right of=7]  {};
        \node[state]  (9) [right of=8, node distance=1.7cm]  {$q_n^+$};

        \node (c)  at ($(1)!0.5!(7)$) {};
        \node (t1) at ($(1)!0.5!(2)$) {$\A_1$};
        \node (t2) at ($(4)!0.5!(5)$) {$\A_2$};
        \node (t3) at ($(7)!0.5!(8)$) {$\A_n$};
        \node (t4) [below of=6, node distance=1.15cm]  {};
        \node (t5) [below of=6, node distance=1.65cm]  {};
        \node at ($(6)!.47!(9)$) {\vdots};
        \node at ($(t2)!.47!(t3)$) {\vdots};
        
        \node[state,initial]  (10) [left of=c, node distance=2.5cm] {$q^-$};
        
        \begin{pgfonlayer}{background}
          \path (1.north -| 2.east) + (0.2,0.3)     node (a) {};
          \path (2.south -| 1.west) + (-0.2,-0.3)   node (b) {};
          \path[rounded corners, draw=black,thick] (a)    rectangle (b);
        \end{pgfonlayer}

        \begin{pgfonlayer}{background}
          \path (4.north -| 5.east) + (0.2,0.3)     node (a) {};
          \path (5.south -| 4.west) + (-0.2,-0.3)   node (b) {};
          \path[rounded corners, draw=black,thick] (a)    rectangle (b);
        \end{pgfonlayer}

        \begin{pgfonlayer}{background}
          \path (7.north -| 8.east) + (0.2,0.3)     node (a) {};
          \path (8.south -| 7.west) + (-0.2,-0.3)   node (b) {};
          \path[rounded corners, draw=black,thick] (a)    rectangle (b);
        \end{pgfonlayer}

        \path[thick]
        (2)  edge node {$a$} (3)
        (5)  edge node {$a$} (6)
        (8)  edge node {$a$} (9)
        (3)  edge node[left] {$\Sigma'$} (6)
        (6)  edge node[left] {$\Sigma'$} (t4)
        (t5) edge node[left] {$\Sigma'$} (9)
        (1)  edge[bend right=25, above] node {$a$} (10)
        (4)  edge[bend right=13, above] node {$a$} (10)
        (7)  edge[bend left=25, above] node {$a$} (10)
        (9)  edge[bend right=40, right] node {$\Sigma'$} (3)
        ;
        \draw[thick] (10) to [out=90,in=125,looseness=8] node[above]{$\Sigma'$} (10)
        ;
      \end{tikzpicture}  
      \caption{Construction of the DES $G$ from the \PSpace-hardness part of the proof of Theorem~\ref{thm_spd-ps}}
      \label{fig4}
    \end{figure}  

    We define the specification $T_{spec}=\{(q^-,q_1^+)\}$, and show that $G'$ is strongly periodically D\mbox{-}de\-tect\-able if and only if the intersection $\cap_{i=1}^n L(\A_i)$ is empty. 
        
    Assume that the intersection is empty. A trajectory that never reaches the states of $Q^+$ cannot violate strong periodic D\mbox{-}de\-tect\-ability, because it cannot enter state $q_1^+$ from the specification. Therefore, assume that $G'$ eventually enters a state of $Q^+$. When $G'$ enters a state of $Q^+$, it leaves all states out of $Q^+\cup\{q^-\}$. Thus, let $s\in L^\omega(G')$ be an arbitrary trajectory that enters $Q^+$. Then, $s=s_1s_2$ where $G'$ generates $s_1$ in states outside $Q^+$, and $s_2$ is the part after $G'$ first enters $Q^+$. In this case, $s_1=f(w)\in\{a,b\}^*$, for some $w\in\{0,1\}^*$, and $s_2\in a\{a,b\}^{\omega}$. Then, after generating the first event of $s_2$, the observer of $G'$ is in a set of states consisting of $q^-$ and a strict subset of $Q^+$; indeed, $G'$ cannot transit to all states of $Q^+$ at the same time, since the assumption that the intersection is empty implies that, for every $w\in\Sigma^*$, there exists $i\in\{1,\ldots,n\}$ such that $w \notin L(\A_i)$. 
     
    Let $p_i\in Q^+=\{q_1^+,\ldots,q_n^+\}$ denote the state of $Q^+$ with the minimal index, in which $G'$ {\em cannot be\/} when the $i$th event of $s_2$ is generated. By construction, the cycle on $Q^+$ ensures that $p_i$ periodically alternates among $q_1^+$ and some other states of $Q^+$ when generating $s_2$. Therefore, in the infinite sequence $p_1,p_2,\ldots$, there are infinitely many $j$ such that $p_j=q_1^+$, and hence $q^-$ and $q_1^+$ are periodically distinguished, which shows that $G$ is strongly periodically D\mbox{-}de\-tect\-able.

    On the other hand, assume that the intersection is nonempty, and let $w\in \cap_{i=1}^{n} L(\A_i)$. Then, after generating the string $f(w)a$, the observer of $G'$ reaches the state $\{q^-\} \cup Q^+$. Now, every transition keeps $G'$ in all states of $\{q^-\} \cup Q^+$, and hence it results in a self-loop in the observer of $G'$. However, this self-loop violates strong periodic D\mbox{-}de\-tect\-ability, because it contains both states $q^-$ and $q_1^+$. Therefore, any trajectory $s\in L^\omega(G')$ with $f(w)a$ as its prefix leads to a set of states where the states of $T_{spec}$ can never be distinguished, and hence $G'$ is not strongly periodically D\mbox{-}de\-tect\-able.
  \end{proof}
  
  Two illustrative examples of our construction without the binary encoding can be found in the conference version~\cite{wodes2020}. Here we illustrate the binary encoding on one of the examples. Let $\A_1$ and $\A_2$ be total DFAs over the alphabet $\{0,1\}$ depicted in Fig.~\ref{fig5a}; $L(\A_1)$ consists of strings of odd length, and $L(\A_2)$ of strings of even length. Our construction results in a DES $G'$ depicted in Fig.~\ref{fig5}. Since $L(\A_1) \cap L(\A_2) = \emptyset$, $G'$ is strongly periodically D\mbox{-}de\-tect\-able, which is evident from the observer of $G_1$ depicted in Fig.~\ref{fig7}, where on every trajectory, $q^-$ and $q_1^+$ can be periodically distinguished.
  \begin{figure}
    \centering
    \begin{tikzpicture}[baseline,auto,->,>=stealth,shorten >=1pt,node distance=1.3cm,
      state/.style={ellipse,minimum size=6mm,inner sep=0pt,thick,draw=black,
        initial text=}, every node/.style={font=}]
      \node[state,initial,initial where=above] (a) {$a$};
      \node[state,accepting] (b) [right of=a, node distance=1.8cm] {$b$};
      
      \node[state] (c) [right of=b, node distance=1.8cm] {$c$};
      \node[state,initial,accepting,initial where=above] (d) [right of=c, node distance=1.8cm] {$d$};

      \node (l1) at (-0.6,0.55) {$\A_1$};
      \node (l2) at (3.1,0.55) {$\A_2$};
      
      \begin{pgfonlayer}{background}
        \path (l1.north -| b.east) + (0.2,0.1)     node (c1) {};
        \path (b.south -| l1.west) + (-0.05,-0.35) node (c2) {};
        \path[rounded corners, draw=black,thick] (c1)    rectangle (c2);
      \end{pgfonlayer}

      \begin{pgfonlayer}{background}
        \path (l2.north -| d.east) + (0.2,0.1)     node (c1) {};
        \path (d.south -| l2.west) + (-0.05,-0.35) node (c2) {};
        \path[rounded corners, draw=black,thick] (c1)    rectangle (c2);
      \end{pgfonlayer}

      \path[thick]
      (a) edge[bend left=15] node {$0,1$} (b)
      (b) edge[bend left=15] node {$0,1$} (a)
      (c) edge[bend left=15] node {$0,1$} (d)
      (d) edge[bend left=15] node {$0,1$} (c)
      ;
    \end{tikzpicture}  
    \caption{The DFAs $\A_1$ and $\A_2$ over $\Sigma=\{0,1\}$}
    \label{fig5a}
  \end{figure}
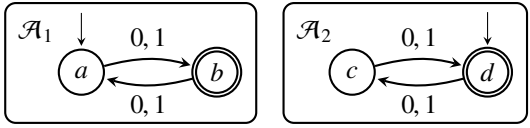  

  \begin{figure}[h]
    \centering
    \begin{tikzpicture}[baseline,auto,->,>=stealth,shorten >=1pt,node distance=1.3cm,
      state/.style={circle,minimum size=6mm,inner sep=0pt,thick,draw=black,
        initial text=}, every node/.style={font=}]
      \node[state,initial,initial where=above] (a) {$a$};
      \node[state] (aa) [above right of=a, node distance=1.6cm] {$e$};
      \node[state] (aaa) [below of=aa, node distance=.8cm] {$f$};
      \node[state] (bb) [below right of=a, node distance=1.6cm] {$h$};
      \node[state] (bbb) [above of=bb, node distance=.8cm] {$g$};
      \node[state,accepting] (b) [right of=a, node distance=2.2cm] {$b$};
      \node[state] (q+1) [right of=b, node distance=1.7cm] {$q_1^+$};
      
      \node[state] (c) [below of=a, node distance=3cm] {$c$};
      \node[state] (cc) [above right of=c, node distance=1.6cm] {$i$};
      \node[state] (ccc) [below of=cc, node distance=.8cm] {$j$};
      \node[state] (dd) [below right of=c, node distance=1.6cm] {$\ell$};
      \node[state] (ddd) [above of=dd, node distance=.8cm] {$k$};
      \node[state,initial,accepting,initial where=above] (d) [right of=c, node distance=2.2cm] {$d$};
      \node[state] (q+2) [right of=d, node distance=1.7cm] {$q_2^+$};

      \node (l1) at (-0.6,1.1) {$\A_1$};
      \node (l2) at (-0.6,-1.9) {$\A_2$};
      \node (temp) at ($(a)!0.5!(c)$) {};      
      \node[state,initial] (q-) [left of=temp, node distance=2.5cm] {$q^-$};
      
      \begin{pgfonlayer}{background}
        \path (l1.north -| b.east) + (0.2,0.1)     node (c1) {};
        \path (bb.south -| l1.west) + (-0.01,-0) node (c2) {};
        \path[rounded corners, draw=black,thick] (c1)    rectangle (c2);
      \end{pgfonlayer}

      \begin{pgfonlayer}{background}
        \path (l2.north -| d.east) + (0.2,0.1)     node (c1) {};
        \path (dd.south -| l2.west) + (-0.01,0) node (c2) {};
        \path[rounded corners, draw=black,thick] (c1)    rectangle (c2);
      \end{pgfonlayer}
 
      \path[thick]
      (b) edge node {$a$} (q+1)
      (d) edge node {$a$} (q+2)
      (a) edge[bend left=15] node[above] {$b$} (aa)
      (a) edge[bend left=15] node[above] {$b$} (aaa)
      (aa) edge[bend left=15] node[above] {$b$} (b)
      (aaa) edge[bend left=15] node[above] {$a$} (b)
      (b) edge[bend left=15] node[above] {$b$} (bb)
      (b) edge[bend left=15] node[above] {$b$} (bbb)
      (bb) edge[bend left=15] node[above] {$b$} (a)
      (bbb) edge[bend left=15] node[above] {$a$} (a)
      (c) edge[bend left=15] node[above] {$b$} (cc)
      (c) edge[bend left=15] node[above] {$b$} (ccc)
      (cc) edge[bend left=15] node[above] {$b$} (d)
      (ccc) edge[bend left=15] node[above] {$a$} (d)
      (d) edge[bend left=15] node[above] {$b$} (dd)
      (d) edge[bend left=15] node[above] {$b$} (ddd)
      (dd) edge[bend left=15] node[above] {$b$} (c)
      (ddd) edge[bend left=15] node[above] {$a$} (c)
      (q+1) edge[bend left=15] node {$a,b$} (q+2)
      (q+2) edge[bend left=15] node {$a,b$} (q+1)
      (a) edge[bend right=15, above] node {$a$} (q-)
      (c) edge[bend left=15, above] node {$a$} (q-)
      (q-) edge[loop above] node{$a,b$} (q-)
      ;
    \end{tikzpicture}  
    \caption{The DES $G'$ with $\Sigma'=\{a,b\}$, where $0$ is encoded as $ba$ and $1$ as $bb$}
    \label{fig5}
  \end{figure}
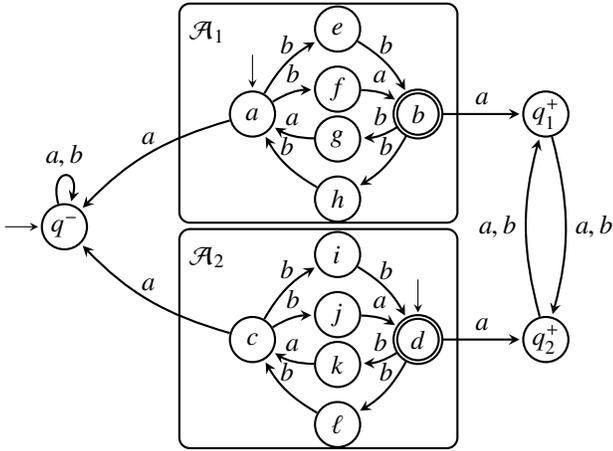  
  
  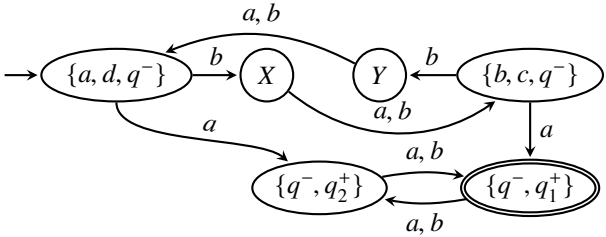
\begin{figure}
    \centering
    \begin{tikzpicture}[baseline,auto,->,>=stealth,shorten >=1pt,node distance=2cm,
      thick, state/.style={ellipse,minimum size=7mm,inner sep=1pt,thick,draw=black,
        initial text=}]
      \node[state,initial] (adq-) {$\{a,d,q^-\}$};
      \node[state] (t1) [right of=adq-, node distance=2cm] {$X$};
      \node[state] (t2) [right of=t1, node distance=1.5cm] {$Y$};
      \node[state] [right of=t2, node distance=2cm] (bcq-) {$\{b,c,q^-\}$};
      \node[state,accepting] [below of=bcq-, node distance=1.5cm] (q-q+1) {$\{q^-,q_1^+\}$};
      \node[state] [left of=q-q+1, node distance=2.8cm] (q-q+2) {$\{q^-,q_2^+\}$};

      \path
      (adq-)  edge[out=270,in=140, looseness=0.6] node[sloped] {$a$} (q-q+2)
      (bcq-)  edge node {$a$} (q-q+1)
      (q-q+2) edge[bend left=10] node {$a,b$} (q-q+1)
      (q-q+1) edge[bend left=10] node {$a,b$} (q-q+2)
      (adq-)  edge node {$b$} (t1)
      (t1)  edge[bend right=37] node[sloped,above] {$a,b$} (bcq-)
      (bcq-)  edge node[above] {$b$} (t2)
      (t2)  edge[bend right=25] node[sloped,above] {$a,b$} (adq-)
      ;
    \end{tikzpicture}  
    \caption{The observer of $G'$; states marked by double circles contain indistinguishable states of $G'$; here $X=\{e,f,k,\ell,q^-\}$ and $Y=\{g,h,i,j,q^-\}$}
    \label{fig7}
  \end{figure}

\subsection*{The case of a single observable event}
  In the previous subsection, we have shown that deciding strong periodic D\mbox{-}de\-tect\-ability is \PSpace-complete for DESs with at least two observable events. We now show that the problem is still more difficult than its non-periodic counterpart even for DESs having only a single observable event.

  \begin{thm}\label{np-c}
    Deciding strong periodic D\mbox{-}de\-tect\-ability for DESs with a single observable event is \NP-complete.
  \end{thm}
  \begin{proof}
    Consider a DES with a single observable event $\{a\}$. If the DES has unobservable events, we can eliminate them as follows. First, we replace each unobservable transition, i.e., a transition of the form $(p,u,q)$ with $u$ being an unobservable event, by an $\varepsilon$-transition $(p,\varepsilon,q)$. Then, we use the standard technique to eliminate $\varepsilon$-transitions~\cite{Hop2007}. This elimination results in a DES and can be done in polynomial time. Therefore, without loss of generality, we may assume that the DES is of the form $G=(Q,\{a\},\delta,I)$.
  
    A proof that we can decide strong periodic D\mbox{-}de\-tect\-ability of $G$ in nondeterministic polynomial time uses a so-called {\em fast matrix multiplication technique}. The basic idea of this technique is to represent the transition function $\delta$ of $G$ as a binary matrix $M$, where $M[i,j]=1$ if and only if there is a transition from state $i$ to state $j$ in $G$. Then, for $r\ge 1$, $M^r$ represents the reachability in $G$ under the string $a^r$. Furthermore, using the fact that $M^2=M\times M$, $M^4 = M^2\times M^2$, etc., we can compute $M^r$ by $O(\log r)$ matrix multiplications, each multiplication in polynomial time. For more details and examples on this technique, we refer to Masopust~\cite{Masopust2018}.

    Assume that $G$ has $n$ states. Then, the observer of $G$ consists of a sequence of $k$ states followed by a cycle consisting of $\ell$ states, that is, the language of $G$ is $a^k (a^{\ell})^*$. Since the number of states of the observer of $G$ is at most $2^n$, $k+\ell \le 2^n$.
 
    Now, $G$ is strongly periodically D\mbox{-}de\-tect\-able if and only if there is a state $X\subseteq Q$ in the cycle of the observer of $G$ (we assume that the observer is constructed by the standard subset construction~\cite{Hop2007}) that is disjoint from the specification, that is, $X\cap T_{spec}=\emptyset$. Indeed, to check whether $X\cap T_{spec}=\emptyset$ can be done in polynomial time. It remains to show how to find $X$ in polynomial time. This means to find $m\le 2^n$ such that $\delta(I,a^{2^n+m})\cap T_{spec}=\emptyset$. However, an NP algorithm can guess $m$ in binary and verify the guess in polynomial time by computing $\delta(I,a^{2^n+m})$ using the fast matrix multiplication, cf. Masopust~\cite{Masopust2018} for more details and an example.

    To prove \NP-hardness, we use the construction of Stockmeyer and Meyer~\cite{StockmeyerM73} encoding a boolean formula in 3CNF in the form of a unary NFA.\footnote{A boolean formula is built from propositional variables, operators conjunction, disjunction, and negation, and parentheses. A formula is satisfiable if there is an assignment of $1$ ({\tt true}) and $0$ ({\tt false}) to its variables making it {\tt true}. A literal is a variable or its negation. A clause is a disjunction of literals. A formula is in conjunctive normal form (CNF) if it is a conjunction of clauses; e.g., $\varphi = (x\lor y \lor z) \land (\neg x\lor y \lor z)$ is a formula in CNF with two clauses $x\lor y \lor z$ and $\neg x\lor y \lor z$. If every clause has at most three literals, the formula is in 3CNF. Given a formula in 3CNF, 3CNF satisfiability asks whether the formula is satisfiable; e.g., $\varphi$ is satisfiable for $(x,y,z)=(0,1,0)$. 3CNF satisfiability is \mbox{\NP-complete}.} For an illustration, the reader may follow Example~\ref{ex2} in parallel with the proof.
  
    Let $\varphi$ be a formula in 3CNF with $n$ variables and $m$ clauses, and let $C_k$ be the set of literals in the $k$th clause, $k=1,\ldots,m$. The assignment to the variables is represented by a binary vector of length $n$. Let $p_1,\ldots, p_n$ be the first $n$ prime numbers. For a natural number $z$ congruent with 0 or 1 modulo $p_i$, for every $i=1,\ldots,n$, we say that $z$ satisfies $\varphi$ if the assignment 
    $
      (z \bmod p_1, z \bmod p_2,\ldots, z \bmod p_n) \in \{0,1\}^{n}
    $
    satisfies $\varphi$. 
    Let $\A_0$ be an NFA recognizing the language of the expression
    $
      \bigcup_{i=1}^{n} \bigcup_{j=2}^{p_i-1} 0^j\cdot (0^{p_i})^*
    $,
    that is, $L(\A_0) = \{ 0^z \mid \exists k \le n, z \not\equiv 0 \bmod p_k \text{ and } z \not\equiv 1 \bmod p_k \}$ is the set of all natural numbers that {\em do not\/} encode an assignment to the variables.

    For each $C_k$, we construct an NFA $\A_k$ such that if $0^z \in L(\A_k)$ and $z$ is an assignment, then $z$ does not assign 1 ({\tt true}) to any literal in $C_k$; e.g., if $C_k = \{x_{r}, \neg x_{s}, x_{t}\}$, $1 \le  r,s,t \le n$ and $r,s,t$ are distinct, let $z_k$ be the unique integer such that $0\le z_k < p_rp_sp_t$, $z_k \equiv 0 \bmod p_r$, $z_k \equiv 1 \bmod p_s$, and $z_k \equiv 0 \bmod p_t$. Then
    $
      L(\A_k) = 0^{z_k} \cdot (0^{p_rp_sp_t})^*
    $. 
    
    Now, $\varphi$ is satisfiable if and only if there exists $z$ such that $z$ encodes an assignment to $\varphi$ and $0^z \notin L(\A_k)$ for all $1\le k \le m$, which is if and only if $L(\A_0) \cup \bigcup_{k=1}^{m} L(\A_k) \neq 0^*$. The construction of all the automata $\A_0,\A_1,\ldots,\A_k$ can be done in polynomial time~\cite{StockmeyerM73}. 
    
    Let $\A$ denote the NFA obtained by taking the automata $\A_0$, $\A_1,\ldots,\A_k$ as a single NFA, and let $p = \Pi_{i=1}^{n} p_i$. If $z$ encodes an assignment to $\varphi$, then so does $z+cp$ for any natural $c$: if $z \equiv x_i \bmod p_i$, then $z + c p \equiv x_i \bmod p_i$, for every $1\le i\le n$, as well. Thus, if  $0^z \notin L(\A_k)$ for all $k$, then $0^z (0^p)^* \cap L(\A) = \emptyset$. Since both languages are infinite, the minimal DFA recognizing $L(\A)$ must have a nontrivial cycle alternating between marked and non-marked states, and hence the same holds for the observer of $\A$.

    We now show that $\varphi$ is satisfiable if and only if $\A$ is strongly periodically D\mbox{-}de\-tect\-able with respect to the specification $T_{spec}$ consisting of all pairs of states, where the states come from two different automata $\A_i$ and $\A_j$, $i\neq j$, and at least one state is marked in its automaton.
    
    Assume that $\varphi$ is satisfiable. As shown above, this is if and only if $L(\A)\neq 0^*$. We have further shown that $L(\A)$ is infinite and that the observer of $\A$ consists of a single trajectory with a non-marked state, $X$, in its cycle part, i.e., $X$ contains only non-marked states of the automata $\A_0,\ldots,\A_k$. Since $T_{spec}$ consists of pairs of states of two different automata $\A_i$ and $\A_j$, $0\le i\neq j\le k$, with at least one state marked in its automaton, we have that $(X\times X) \cap T_{spec}=\emptyset$. Therefore, $\A$ is strongly periodically D\mbox{-}de\-tect\-able with respect to the specification $T_{spec}$.
        
    On the other hand, assume that $\varphi$ is not satisfiable. As shown above, this is if and only if $L(\A)=0^*$. But then every state of the observer of $\A$ must be marked, i.e., every state, $X$, of the observer contains a marked state of some $\A_i$, $0\le i\le k$, and hence $(X\times X) \cap T_{spec} \neq \emptyset$. Therefore, $\A$ is not strongly periodically D\mbox{-}de\-tect\-able with respect to $T_{spec}$.
  \end{proof}
    
  The following example illustrates the construction.
  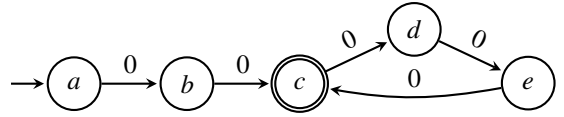
\begin{figure}
    \centering
    \begin{tikzpicture}[baseline,auto,->,>=stealth,shorten >=1pt,node distance=1.5cm,thick, state/.style={ellipse,minimum size=7mm,inner sep=1pt,thick,draw=black,initial text=}]
    \node[state,initial]  (1) {$a$};
    \node[state]  (2) [right of=1] {$b$};
    \node[state,accepting]  (3) [right of=2]  {$c$};
    \node[state]  (4) [above right=.2cm and 1cm of 3]  {$d$};
    \node[state]  (5) [below right=.2cm and 1cm of 4]  {$e$};
    \path
      (1) edge node {$0$} (2)
      (2) edge node {$0$} (3)
      (3) edge node[sloped,above] {$0$} (4)
      (4) edge node[sloped,above] {$0$} (5)
      (5) edge[bend left=10] node[above] {$0$} (3)
    ;
    \end{tikzpicture}
    \caption{Automaton $\A_{1,0}$}
    \label{figA11}
  \end{figure}

  \begin{exmp}\label{ex2}
    Let $\varphi_1 = (x\lor y) \land (\neg x\lor y)$ and $\varphi_2 = x \land \neg x$. Obviously, $\varphi_1$ is satisfiable and $\varphi_2$ is not. For both formulae, we can construct the unary automata $\A_1$ and $\A_2$, respectively, and show that $\A_1$ is strongly periodically D\mbox{-}de\-tect\-able while $\A_2$ is not. In this paper, we construct only the automaton $\A_1$; the construction of the automaton $\A_2$ can be found in the conference version~\cite{wodes2020}.
    
    The formula $\varphi_1 = (x\lor y) \land (\neg x\lor y)$ has two variables, and therefore we set $p_1=2$ and $p_2=3$, the first two prime numbers. The automaton $\A_{1,0}$, depicted in Fig.~\ref{figA11}, recognizes the language $0^2 (0^3)^*$ of all strings that do not encode the assignment to $\varphi_1$.
    Since $\varphi_1$ consists of two clauses, we further construct two automata: $\A_{1,1}$ recognizing the language $(0^6)^*$, and $\A_{1,2}$ recognizing the language $0^3 (0^6)^*$; the reader can verify that if $0^z \in L(\A_{1,1})$ and $z$ is an assignment, then $z$ assigns {\tt true} neither to $x$ nor to $y$, and if $0^z \in L(\A_{1,2})$ and $z$ is an assignment, then $z$ assigns {\tt true} to $x$ (that is, it assigns {\tt false} to the literal $\neg x$) and {\tt false} to $y$. The automata are depicted in Figs.~\ref{figA12} and~\ref{figA13}, respectively.
    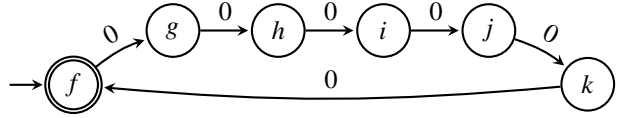
\begin{figure}
      \centering
      \begin{tikzpicture}[baseline,auto,->,>=stealth,shorten >=1pt,node distance=1.4cm,thick, state/.style={ellipse,minimum size=7mm,inner sep=1pt,thick,draw=black,initial text=}]
      \node[state,initial,accepting] (1) {$f$};
      \node[state]  (2) [above right=.2cm and .8cm of 1] {$g$};
      \node[state]  (3) [right of=2]  {$h$};
      \node[state]  (4) [right of=3]  {$i$};
      \node[state]  (5) [right of=4]  {$j$};
      \node[state]  (6) [below right=.2cm and .8cm of 5]  {$k$};
      \path 
        (1) edge[bend left=10] node[sloped,above] {$0$} (2)
        (2) edge node {$0$} (3)
        (3) edge node {$0$} (4)
        (4) edge node {$0$} (5)
        (5) edge[bend left=10] node[sloped,above] {$0$} (6)
        (6) edge[bend left=5] node[above] {$0$} (1)
      ;
      \end{tikzpicture}
      \caption{Automaton $\A_{1,1}$}
      \label{figA12}
    \end{figure}
    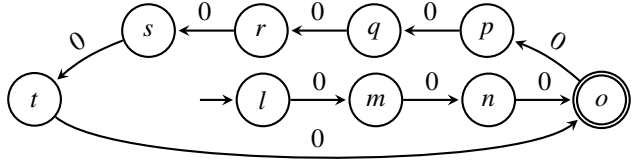
\begin{figure}
      \centering
      \begin{tikzpicture}[baseline,auto,->,>=stealth,shorten >=1pt,node distance=1.5cm,thick, state/.style={ellipse,minimum size=7mm,inner sep=1pt,thick,draw=black,initial text=}]
      \node[state,initial] (l) {$l$};
      \node[state] (m) [right of=l] {$m$};
      \node[state] (n) [right of=m] {$n$};
      \node[state,accepting] (1) [right of=n] {$o$};
      \node[state]  (2) [above left=.4cm and 1cm of 1] {$p$};
      \node[state]  (3) [left of=2]  {$q$};
      \node[state]  (4) [left of=3]  {$r$};
      \node[state]  (5) [left of=4]  {$s$};
      \node[state]  (6) [below left=.4cm and 1cm of 5]  {$t$};
      \path 
        (l) edge node {$0$} (m)
        (m) edge node {$0$} (n)
        (n) edge node {$0$} (1)
        (1) edge[bend right=10] node[sloped,above] {$0$} (2)
        (2) edge node[above] {$0$} (3)
        (3) edge node[above] {$0$} (4)
        (4) edge node[above] {$0$} (5)
        (5) edge[bend right=10] node[sloped,above] {$0$} (6)
        (6) edge[out=320,in=220, looseness=0.4] node {$0$} (1)
      ;
      \end{tikzpicture}
      \caption{Automaton $\A_{1,2}$}
      \label{figA13}
    \end{figure}%
    The specification
    \begin{align*}
        T_{spec} & = 
        \{c\}\times \{f,g,h,i,j,k,l,m,n,o,p,q,r,s,t\}\\
        & \cup
        \{a,b,c,d,e,l,m,n,o,p,q,r,s,t\}\times \{f\}\\
        & \cup
        \{a,b,c,d,e,f,g,h,i,j,k\}\times\{o\}\,.
    \end{align*}
    Let $\A_1$ be the NFA consisting of the automata $\A_{1,0}$, $\A_{1,1}$, and $\A_{1,2}$. The observer of $\A_{1}$ is depicted in Fig.~\ref{figA1}.
    \begin{figure}
      \centering
      \begin{tikzpicture}[baseline,auto,->,>=stealth,shorten >=1pt,node distance=2.4cm,thick, state/.style={scale=0.7,ellipse,minimum size=6mm,inner sep=2pt,thick,draw=black,initial text=}]
      \node[state,initial,accepting] (l) {$\{a,f,l\}$};
      \node[state] (m) [right of=l] {$\{b,g,m\}$};
      \node[state,accepting] (n) [right of=m] {$\{c,h,n\}$};
      \node[state,accepting] (1) [right of=n] {$\{d,i,o\}$};
      \node[state] (2) [above left=.5cm and .1cm of 1] {$\{e,j,p\}$};
      \node[state,accepting] (3) [left of=2]  {$\{c,k,q\}$};
      \node[state,accepting] (4) [left of=3]  {$\{d,f,r\}$};
      \node[state] (5) [left of=4] {$\{e,g,s\}$};
      \node[state,accepting] (6) [below left=.5cm and .1cm of 5] {$\{c,h,t\}$};
      \path 
        (l) edge node {$0$} (m)
        (m) edge node {$0$} (n)
        (n) edge node {$0$} (1)
        (1) edge[bend right=10] node[sloped,above] {$0$} (2)
        (2) edge node[above] {$0$} (3)
        (3) edge node[above] {$0$} (4)
        (4) edge node[above] {$0$} (5)
        (5) edge[bend right=10] node[sloped,above] {$0$} (6)
        (6)  edge[out=320,in=220, looseness=0.4] node[above] {$0$} (1)
      ;
      \end{tikzpicture}
      \caption{The observer of the NFA $\A_{1}$; states marked by a double circle contain indistinguishable states of $\A_{1}$}
      \label{figA1}
    \end{figure}
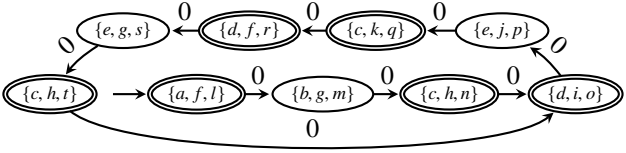
    The reader can see that the observer contains a cycle where the state $\{e,j,p\}$ appears periodically. Since this state does not contain any pair from the specification $T_{spec}$, the NFA $\A_1$ is strongly periodically D\mbox{-}de\-tect\-able as claimed.
  \hfill$\diamond$
  \end{exmp}

  As already pointed out, the observer of a DES with a single observable event consists of a single trajectory ending with a cycle. We now formulate the consequences of this observation.
  
  \begin{thm}\label{thm6}
    For a DES with a single observable event, deciding weak detectability coincides with deciding strong detectability, and deciding weak periodic detectability coincides with deciding strong periodic detectability. The same holds true for D\mbox{-}de\-tect\-ability.
  \end{thm}
  \begin{proof}
    Let $G$ be a DES with $n$ states and a single observable event $\Sigma_o=\{a\}$. Then the observer of $G$ consists of a sequence of $k$ states followed by a cycle consisting of $\ell$ states, that is, the language of $G$ is $a^k (a^{\ell})^*$ with $k+\ell \le 2^n$. The existence of the cycle follows from the two assumptions on Page~\pageref{AS}. Since the observer of $G$ consists of a single infinite trajectory, deciding strong (periodic) detectability coincides with deciding weak (periodic) detectability; the same holds when we replace detectability with D\mbox{-}de\-tect\-ability.
  \end{proof}

\section{A Tractable Case}\label{secRpoNFAs}
  In the previous sections, we have shown that deciding strong periodic D\mbox{-}de\-tect\-ability is a difficult problem for DES modeled by NFAs. In this section, we discuss a special case of systems for which the verification of strong periodic D\mbox{-}de\-tect\-ability is polynomial.
  
  Looking for a class of DESs, for which the problem is tractable, we first inspect the proof of Theorem~\ref{thm_spd-ps}. This reveals that the proof is based on the intersection emptiness problem for DFAs that was shown by Kozen~\cite{Kozen77} to be \PSpace-complete. His proof heavily relies on DFAs with cycles. Allowing only self-loops instead of cycles in the DFAs makes the problem easier~\cite{arxiv2019}. Therefore, we consider DESs modeled by NFAs where all cycles in the transition graph are only self-loops. A self-loop may be added to any state, and hence the NFA may fulfill the deadlock-free requirement. Such NFAs recognize a strict subclass of regular languages that are strictly included in {\em star-free languages}~\cite{BrzozowskiF80,mfcs16:mktmmt_full}. Star-free languages are languages definable by {\em linear temporal logic\/}, which is a logic widely used as a specification language in automated verification. 
  
  A practical motivation for such systems comes from the following observation. Every infinite trajectory in a system describes a task that is possibly repeated ad infinitum. Indeed, every task is a finite sequence of events, though for the modeling purposes some repetitions of subtasks may be modeled as cycles. In some cases and on same level of abstraction, these internal cycles could be seen as self-loops. This results in a system with only self-loops. Now we need to model the situation that the task has been finished and the whole process can be restarted. This can be done by the repeated generation of a special event telling the system to restart the specific task.

  Let $\A=(Q,\Sigma,\delta,I,F)$ be an NFA. The reachability relation $\le$ on the state set $Q$ is defined by $p\le q$ if there is $w\in \Sigma^*$ such that $q\in \delta(p,w)$. The NFA $\A$ is {\em restricted partially ordered (rpoNFA)\/}  if the reachability relation $\le$ is a partial order and $\A$ is self-loop deterministic in the sense that the pattern of Fig.~\ref{fig_bad_pattern} does not appear. Formally, for every state $q$ and every event $a$, if $q\in \delta(q,a)$ then $\delta(q,a) = \{q\}$.
  \begin{figure}
    \centering
    \begin{tikzpicture}[baseline,->,>=stealth,auto,shorten >=1pt,node distance=2cm,thick,
      state/.style={circle,minimum size=5mm,inner sep=2pt,thick,draw=black,initial text=}]
      \node[state]  (a) {};
      \node[state]  (aa) [right of=a]  {};
      \path
        (a) edge[loop above] node {$a$} (a)
        (a) edge node {$a$} (aa)
        ;
    \end{tikzpicture}
    \caption{The forbidden pattern of rpoNFAs}
    \label{fig_bad_pattern}
  \end{figure}
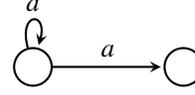
  We now formally define so-called rpoDES. The name comes from restricted partially ordered DES.
  \begin{defn}
    Let $G=(Q,\Sigma,\delta,I)$ be a DES with $\Sigma_o$ being the set of observable events. Let $P\colon \Sigma^* \to \Sigma_o^*$ be the corresponding projection. We say that $G$ is an rpoDES if the NFA $P(G)=(Q,\Sigma_o,\delta',I)$ obtained from $G$ by replacing every transition $(p,a,q)$ by $(p,P(a),q)$, and by eliminating the $\varepsilon$-transitions~\cite{Hop2007} is an rpoNFA.
  \end{defn}
  
  Notice that $P(G)$ is an NFA that can be constructed from $G$ in polynomial time~\cite{Hop2007}. Therefore, the question whether a DES is an rpoDES is decidable in polynomial time.
  
  What do we know about rpoDES? Deciding weak (periodic) detectability for rpoDESs is \PSpace-complete~\cite{Masopust2018b}, and hence so is deciding weak (periodic) D\mbox{-}de\-tect\-ability. 
  
  We now show that the complexity of deciding strong periodic D\mbox{-}de\-tect\-ability for rpoDESs coincides with the complexity of deciding strong D\mbox{-}de\-tect\-ability.

  \begin{thm}\label{thmRpoNFAs}
    Deciding strong periodic D\mbox{-}de\-tect\-ability for rpoDESs is \NL-complete.
  \end{thm}
  \begin{proof}
    To prove the theorem, we show that the observer of an rpoDES is a partially ordered DFA. Then, since there are no nontrivial cycles in the observer, strong periodic D\mbox{-}de\-tect\-ability coincides with strong D\mbox{-}de\-tect\-ability. Theorem~\ref{thm_sdnl-c} then finishes the proof.
  
    Let $\A=(Q,\Sigma_o,\delta,I,F)$ be an rpoNFA. We show that the DFA $\D$ computed from $\A$ by the standard subset construction (i.e., the observer) is partially ordered.
    To this aim, let $X=\{p_1,\ldots,p_n\}$ with $p_i < p_j$ for $i<j$ be a state of $\D$, and let $w\in\Sigma^*$ be a nonempty string such that $\delta_{\D}(X,w) = X$. First, we show that $\delta(p_i,w) = \{p_i\}$ for all $i$. For the sake of contradiction, let $1\le k\le n$ be the minimal integer such that $\delta(p_k,w) \neq \{p_k\}$. Since $X = \delta(X,w) = \cup_{i=1}^{n} \delta(p_i,w)$, $\delta(p_i,w) = \{p_i\}$ for all $i<k$, and $p_k < p_i \le \delta(p_i,w)$ for all $i>k$, we have that $p_k \notin \cup_{i=1}^{n} \delta(p_i,w) = X$, which is a contradiction.\footnote{By $p_i \le \delta(p_i,w)$ we mean $p_i\le q$ for every $q\in\delta(p_i,w)$.} Therefore, $p_k \in \delta(p_k,w)$, and the definition of rpoNFAs implies that every event of $w$ is in a self-loop in state $p_k$. Because rpoNFAs have no choice between staying in the state and leaving it under the same event, $\delta(p_k,w) = \{p_k\}$. Thus, for $i=1,\ldots,n$, $\delta(p_i,a) = \{p_i\}$ for every event $a$ occurring in $w$. Consequently, for any state $Y$ of $\D$ and any strings $w_1$ and $w_2$, if $\delta(X,w_1) = Y$ and $\delta(Y,w_2) = X$, the previous argument gives that $X=Y$, and hence $\D$ is partially ordered.
  \end{proof}

  Finally, we have the following corollary of Theorem~\ref{thm_sdnl-c}.
  \begin{cor}\label{cor1}
    Deciding whether an rpoDES with a single observable event is strongly (D\mbox{-})detect\-able is \NL-complete.
  \end{cor}
  \begin{proof}
    This result is an immediate consequence of the proof of Theorem~\ref{thm_sdnl-c}, since in the hardness part we actually construct a unary rpoNFA, cf. Fig.~\ref{fig1}. 
  \end{proof}

\section{Conclusions}
  In this paper, we answered the open question concerning the complexity of deciding whether a DES satisfies strong periodic D\mbox{-}de\-tect\-ability, and provided a full complexity picture of this problem. Since the results for DES are mainly negative, we also discussed a class of DESs, so-called rpoDESs, for which the complexity of deciding strong (periodic) D\mbox{-}de\-tect\-ability is tractable.

\ack{Supported by the Ministry of Education, Youth and Sports under the INTER-EXCELLENCE project LTAUSA19098, by the Czech Science Foundation project GC19-06175J, by IGA PrF 2020 019, and by RVO~67985840.}

\bibliographystyle{plain}
\bibliography{mybib}

\end{document}